\documentclass[11pt]{amsart}

\usepackage[T1]{fontenc}
\usepackage[utf8]{inputenc}
\usepackage{lmodern}
\usepackage{amsmath,amssymb,amsthm,mathtools}
\usepackage{stmaryrd}
\usepackage{microtype}
\usepackage[colorlinks=true,linkcolor=blue,citecolor=blue,urlcolor=blue]{hyperref}
\usepackage[nameinlink,capitalise,noabbrev]{cleveref}
\usepackage[a4paper,margin=1in]{geometry}
\usepackage{float}
\usepackage{tikz}
\usepackage{tikz-cd}
\usetikzlibrary{arrows.meta,positioning,calc}

\newtheorem{theorem}{Theorem}[section]
\newtheorem{proposition}[theorem]{Proposition}

\theoremstyle{definition}
\newtheorem{definition}[theorem]{Definition}
\newtheorem{example}[theorem]{Example}

\theoremstyle{remark}
\newtheorem{remark}[theorem]{Remark}

\newcommand{\Set}{\mathbf{Set}}

\newcommand{\TBox}{\mathsf{T}}
\newcommand{\ABox}{\mathsf{A}}
\newcommand{\PBox}{\mathsf{P}}
\newcommand{\OBox}{\mathsf{O}}
\newcommand{\KB}{\mathsf{KB}}
\newcommand{\Ctx}{\mathcal{U}}
\newcommand{\Sem}[1]{\llbracket #1 \rrbracket}
\newcommand{\State}{\Sigma}

\title[TAPO-Description Logic for Information Behavior]{TAPO-Description Logic for Information Behavior: Refined OBoxes, Inference, and Categorical Semantics}
\author{Takao Inou\'{e}}
\address{Faculty of Informatics, Yamato University, Osaka, Japan}
\email{inoue.takao@yamato-u.ac.jp}
\thanks{Personal Email: takaoapple@gmail.com. I prefer my personal email address for correspondence.}
\date{April 22, 2026}


\begin{document}

\begin{abstract}
This paper develops a refined version of TAPO-description logic for the analysis of information behavior. The framework is treated not as a single homogeneous object logic, but as a layered formalism consisting of a static descriptive layer (TBox/ABox), a procedural layer (PBox), and an oracle-sensitive layer (OBox). To make this architecture mathematically explicit, we introduce a metalevel guard-judgment layer governing procedural branching and iteration. On this basis we formulate a core inference system for TAPO-description logic, covering static TBox/ABox reasoning, guarded procedural transition in the PBox, and validated external import in the OBox. We then give a categorical semantics for the resulting framework and indicate its sheaf-theoretic refinement. The theory is illustrated by examples of information-seeking behavior, including simple search behavior and review-sensitive ordering behavior in a curry restaurant. The aim is to treat not only static knowledge representation but also hesitation, external consultation, and action-guiding update within a unified logical setting.\end{abstract}

\maketitle

\tableofcontents

\section{Introduction}
Description logic has long provided a mathematically precise language for structured knowledge representation. Its classical architecture is centered around a terminological component and an assertional component, usually called the TBox and the ABox. This framework has proved powerful for ontology engineering and semantic technologies, but many actual information environments involve more than static concepts and assertions. They also involve procedures, controlled update, and access to information obtained from outside the current knowledge state.

The author previously introduced TAPO-description logic as a four-layered extension of the standard picture, consisting of a TBox, an ABox, a PBox, and an OBox \cite{InoueTAPO}. In that earlier paper, the emphasis was on the conceptual motivation for enriching description logic by procedural and oracle-based components. The purpose of the present paper is to push that program in a more explicitly mathematical direction. In particular, we aim to isolate a clean inferential skeleton that is robust enough to support later categorical, sheaf-theoretic, and topos-theoretic refinement.

Recent work has shown that ordinary description logics, especially systems around $\mathcal{ALC}$ and general TBoxes, admit meaningful category-theoretic semantics \cite{LeDuc2021,BrieulleLeDucVaillant2022,LeDucBrieulle2025}. There are also broader categorical approaches to knowledge representation that lie close to description logic, for example the bicategorical relational framework of Patterson \cite{Patterson2017}. On a different but related side, recent NLP-oriented work on extracting formal models from natural-language specifications has made explicit the importance of conditional and action-oriented intermediate structure \cite{GhoshEtAl2022SpecNFS}. These developments suggest that the passage from set-theoretic semantics to structural semantics is natural and fruitful. Nevertheless, most direct semantics for description logic still concentrate on the static TBox/ABox core, while parsing-oriented approaches do not by themselves provide an intrinsic treatment of procedural knowledge. The present paper is intended as a next step: not a survey of those directions, but a categorical and proof-theoretic continuation of the TAPO program.

From the viewpoint of information behavior, an additional nearby motivation comes from the literature on browsing. Chang and Rice treat browsing not as a weak substitute for directed search, but as a multidimensional and context-sensitive information activity shaped by motivation, interface structure, cognitive state, and available resources \cite{ChangRice1993}. Related work by Chang emphasizes that browsing interfaces and patterns of information organization can materially influence subsequent user action \cite{ChangAmazon2001}, while the broader study of access and browsing by Rice, McCreadie, and Chang argues that observable query--answer episodes are often only one visible portion of a larger chain of informational activity \cite{RiceMcCreadieChang2001}. This perspective is strongly compatible with the passage from TBox/ABox-structured to TAPO-structured description logic: once browsing, hesitation, reformulation, and externally mediated confirmation are taken seriously, procedural and oracle-sensitive layers cease to be optional embellishments and become mathematically natural components of the theory.

The central thesis of the present paper is that these additional layers should be treated as genuine mathematical structure. The PBox should not be reduced to an informal programming metaphor; it should be represented by programs, guarded evaluations, and composable state transformations over localized knowledge states. Likewise, the OBox should not remain a vague placeholder for ``external input.'' It should be strengthened into a controlled interface equipped with admissible queries, contextual responses, compatibility with restriction, and explicit criteria for trusted and certified import. This strengthening of the OBox, together with the introduction of an explicit inferential layer for procedural control, is one of the main ways in which the present paper goes beyond the original formulation in \cite{InoueTAPO}.

A point that is essential for the present paper is that TAPO-DL should not be read as a single homogeneous object language. Rather, it is a layered formalism. The TBox and ABox form the static descriptive layer; the PBox provides a procedural layer with constructs such as \textbf{if}, \textbf{while}, and sequencing; and the OBox provides an oracle-sensitive layer for controlled external import. The judgments used in connection with these procedural and oracle-sensitive constructs do not themselves belong to the underlying description-logic language. They are metalevel relations specifying how guards are evaluated, how procedures are executed, and how externally obtained information is admitted into the current state. In concrete applications to information behavior, such judgments may be realized by a system, by a user, or by an oracle-mediated interaction process. This separation of layers is one of the essential features of TAPO-DL as developed here.

To make this separation mathematically explicit, the paper introduces a guard-judgment layer for procedural branching and iteration. Instead of leaving \textbf{if} and \textbf{while} at the level of informal pseudocode, we give them proof-theoretic transition rules mediated by guard evaluation. This does not amount to a full internalization of procedural behavior into the object language of description logic; rather, it records at the metalevel how the procedural layer interacts with the static DL layer. In this respect the framework is intentionally hybrid: descriptive formulas remain in the DL layer, while procedural control is governed by judgments that are computationally meaningful and theoretically analyzable.

The paper is organized as follows. In \cref{sec:prelim} we recall the refined TAPO-DL architecture and fix the localized notation used throughout. In \cref{sec:obox} we sharpen the OBox by introducing admissible queries and responses, trust levels, certification data, and validation policies. In \cref{sec:proofsys} we present a core inference system for the refined framework, including the guard-judgment layer for the PBox. The categorical semantics is then developed in \cref{sec:tapo-objects,sec:ta,sec:pbox,sec:sheaf}. We record a schematic soundness statement together with the associated completeness problem in \cref{sec:soundcomplete}, and only after that turn to information-behavior examples in \cref{sec:examples}, ranging from simple search behavior to review-sensitive ordering behavior in a curry restaurant. Finally, \cref{sec:future} records further directions.

\section{Refined TAPO-description logic}\label{sec:prelim}
\subsection{Signature, concepts, and contexts}
Let the TAPO-signature be
\[
\Sigma_{\mathrm{sig}}=(N_C,N_R,N_I,\Ctx),
\]
where $N_C$ is a set of concept names, $N_R$ a set of role names, $N_I$ a set of individual names, and $\Ctx$ is a collection of contextual domains. We regard $\Ctx$ as a poset under refinement, writing $V\leq U$ when $V$ is a refinement of $U$.

Concept expressions are generated by the usual $\mathcal{ALC}$ grammar
\[
C ::= \top \mid \bot \mid A \mid C\sqcap D \mid C\sqcup D \mid \neg C \mid \exists r.C \mid \forall r.C,
\]
with $A\in N_C$ and $r\in N_R$.

\subsection{The TAPO architecture}
A TAPO-knowledge state over a context $U\in\Ctx$ consists of four layers:
\[
\KB(U)=(\TBox,\ABox_U,\PBox_U,\OBox_U).
\]
Here:
\begin{itemize}
    \item $\TBox$ records terminological data, such as concept inclusions and structural axioms;
    \item $\ABox_U$ records contextual assertions valid over $U$;
    \item $\PBox_U$ records admissible procedures over the localized state $(\TBox,\ABox_U)$;
    \item $\OBox_U$ records admissible modes of interaction with external information sources over $U$.
\end{itemize}
The original logical motivation for this four-layer structure is given in \cite{InoueTAPO}.

\subsection{Localization and restriction}
Restriction along $V\leq U$ is written
\[
\rho_{UV}:\KB(U)\to \KB(V).
\]
At the skeletal level we require functoriality of restriction and compatibility with the four-layer TAPO decomposition.

\begin{definition}[Presheaf of TAPO-states]
A \emph{presheaf of TAPO-states} on $\Ctx$ is a contravariant functor
\[
\KB:\Ctx^{\mathrm{op}}\to \Set
\]
whose value $\KB(U)$ is the set of TAPO-knowledge states over $U$.
\end{definition}

\begin{remark}
The point of using a presheaf already at this stage is that TBox/ABox/P/O-data need not be globally settled. Context dependence, partiality of information, and later gluing are built in from the start.
\end{remark}

\section{The OBox as oracle interface and controlled external import}\label{sec:obox}
In the present reorganization, the oracle layer is refined before the proof theory and before the categorical semantics are unfolded in detail. This reflects the view that the strengthened OBox is one of the main formal innovations of the paper.

The oracle layer is the place where the present paper most decisively strengthens the earlier logical skeleton from \cite{InoueTAPO}. In the original paper, the OBox was presented at a minimal level by an externally justified transition relation
\[
\Sem{\OBox}\subseteq (\TBox\times \ABox)\times (\TBox\times \ABox).
\]
Here we retain that intuition but replace it with a more structured semantics.

\subsection{The strengthened OBox}
\begin{definition}[Strengthened oracle frame]
Let $X_U=(\State_U,\PBox_U,\OBox_U)$ with $\State_U=(\TBox,\ABox_U)$. A \emph{strengthened oracle frame} over $U$ consists of the following data:
\begin{enumerate}
    \item a set $Q_U$ of admissible queries;
    \item a set $R_U$ of admissible responses;
    \item a partial answer map
    \[
    \operatorname{ans}_U:Q_U\dashrightarrow R_U;
    \]
    \item an import map
    \[
    \operatorname{imp}_U:R_U\dashrightarrow \mathcal P(\mathrm{Asrt}(U));
    \]
    \item a preordered set $(L_U,\preceq)$ of trust levels;
    \item a trust assignment
    \[
    \operatorname{tr}_U:R_U\to L_U;
    \]
    \item a certificate space $C_U$ and certificate assignment
    \[
    \operatorname{cert}_U:R_U\to \mathcal P(C_U);
    \]
    \item a chosen trust threshold $\ell_U\in L_U$;
    \item a validation policy
    \[
    \operatorname{val}_U(r,S)\in\{\mathsf{accept},\mathsf{reject},\mathsf{defer}\},
    \]
    defined for $r\in R_U$ and finite certificate sets $S\subseteq C_U$.
\end{enumerate}
\end{definition}

\begin{remark}
Compared with the earlier TAPO-DL paper \cite{InoueTAPO}, the present formulation extends the OBox in four explicit directions: admissible queries and responses are separated; raw oracle answers are distinguished from imported assertions; trust and certification data are recorded; and a validation policy mediates which externally obtained information is actually integrated into the ABox.
\end{remark}

\subsection{Validated import}
\begin{definition}[Validated response]
A response $r\in R_U$ is \emph{validated} if there exists a finite set $S\subseteq \operatorname{cert}_U(r)$ such that
\[
\operatorname{tr}_U(r)\succeq \ell_U
\quad\text{and}\quad
\operatorname{val}_U(r,S)=\mathsf{accept}.
\]
\end{definition}

\begin{definition}[Validated import map]
The \emph{validated import map}
\[
\operatorname{vimp}_U:R_U\dashrightarrow \mathcal P(\ABox_U)
\]
is defined exactly on validated responses and agrees there with the ordinary import map $\operatorname{imp}_U$.
\end{definition}

\begin{definition}[Sound strengthened OBox]
A strengthened oracle frame is \emph{sound} if every validated response imports only assertions compatible with the TBox $\TBox$.
\end{definition}

\begin{proposition}[Certified oracle transition]
A sound strengthened oracle frame determines a partial endomap
\[
\omega_U^{\mathrm{val}}:\State_U\dashrightarrow \State_U
\]
by adjoining the validated imported assertions to $\ABox_U$ whenever defined.
\end{proposition}

\begin{proof}
If a response is validated, then by definition it passes the trust and certification threshold; if the frame is sound, its imported assertions are compatible with $\TBox$. Hence the updated pair $(\TBox,\ABox_U')$ is again a localized knowledge state.
\end{proof}

\begin{remark}
In this way the original transition-relation view of the OBox is recovered as a shadow of a more structured oracle mechanism. The present categorical skeleton therefore keeps the spirit of \cite{InoueTAPO} while making the oracle layer mathematically more explicit.
\end{remark}

\subsection{Restriction compatibility}
\begin{definition}[Restriction-compatible strengthened OBox]
Let $V\leq U$. A strengthened oracle frame over $U$ is \emph{restriction-compatible} if there are induced maps
\[
Q_U\to Q_V,
\qquad
R_U\to R_V,
\qquad
L_U\to L_V,
\qquad
C_U\to C_V,
\]
such that answer maps, trust assignments, certificate assignments, validation policies, and validated import commute with restriction of localized knowledge states.
\end{definition}

\begin{proposition}[Restriction of validated oracle imports]
Suppose the strengthened OBox is restriction-compatible. Then validated oracle imports commute with contextual restriction wherever both sides are defined.
\end{proposition}

\begin{proof}
This is exactly the content of restriction-compatibility: a response, its trust data, its certification data, and its validation status are all transported coherently along restriction, and validated import is required to commute with the restriction maps of the surrounding TAPO-presheaf.
\end{proof}

\subsection{Compositionality of oracle interfaces}
\begin{definition}[Composable strengthened oracle frames]
Two strengthened oracle frames on $\State_U$ are \emph{composable} if the validated output state of the first lies in the domain of the second whenever both are defined.
\end{definition}

\begin{proposition}[Closure under certified composition]
Let $\omega_1$ and $\omega_2$ be composable strengthened oracle frames. If the validation policy for $\omega_2$ is allowed to inspect certificates produced or preserved by $\omega_1$, then the composite oracle transition again admits a natural certified semantics.
\end{proposition}

\begin{proof}
One composes the underlying partial transitions and carries the certification data produced at the first stage to the second-stage validation policy. Soundness is inherited provided each stage accepts only T-compatible imports.
\end{proof}

\begin{remark}
This again goes beyond the original OBox formulation in \cite{InoueTAPO}. There the main point was controlled openness to external information. Here that openness is refined into a compositional interface with provenance-sensitive validation.
\end{remark}

\begin{remark}[How the present OBox extends the original TAPO-DL layer]
In the original TAPO-DL paper \cite{InoueTAPO}, the OBox was introduced mainly as a relation of externally justified state transition. The present paper keeps that basic idea, but extends it in a form better suited to information-behavior analysis. The extension has four visible aspects: the query posed by the agent is separated from the response returned by the external source; imported assertions are distinguished from raw responses; trust and certification are represented explicitly; and a validation policy determines which externally obtained content is actually integrated into the ABox. In this sense the strengthened OBox is not a different oracle layer, but a structured refinement of the original one.
\end{remark}

\subsection{Why the OBox is mathematically nontrivial}
The strengthened OBox is not merely a box labelled ``outside information.'' It carries at least five visible pieces of structure: admissible questions, admissible answers, an import mechanism, trust/certification data, and validation policies compatible with context. From a categorical viewpoint, the oracle layer behaves like a structured interface between an internal knowledge state and a family of external information channels.

\section{A core inference system for TAPO-description logic}\label{sec:proofsys}
The semantic skeleton developed above becomes substantially stronger once one records a corresponding proof-theoretic layer. We therefore introduce a small inference system whose purpose is to make explicit how static entailment, guard evaluation, procedural transition, and oracle-mediated import may be derived inside TAPO-description logic itself. A fuller completeness theorem is left for later refinement; the present calculus is intended as the core proof-theoretic layer compatible with the semantics developed here.

\subsection{Judgments}
We use five kinds of judgments.
\begin{enumerate}
    \item Terminological consequence:
    \[
    \TBox \vdash C\sqsubseteq D.
    \]
    \item Assertional consequence over a context $U$:
    \[
    \State_U \vdash a:C@U,
    \qquad
    \State_U \vdash (a,b):r@U.
    \]
    \item Static guard evaluation:
    \[
    (\State_U,\PBox_U)\Downarrow_g \gamma : \mathbf t,
    \qquad
    (\State_U,\PBox_U)\Downarrow_g \gamma : \mathbf f.
    \]
    \item Procedural transition:
    \[
    (\State_U,\PBox_U)\vdash P: \State_U \rightsquigarrow \State_U'.
    \]
    \item Oracle-mediated validated import:
    \[
    (\State_U,\OBox_U)\vdash q \xRightarrow{V} \State_U'.
    \]
\end{enumerate}
The first two are static descriptive judgments. The third is a static guard judgment attached to the procedural layer and records whether a branching condition is currently evaluated as true or false. The fourth is dynamic and records state transformation by a PBox program. The fifth records validated acquisition of external information under a fixed validation policy $V$.

\subsection{Static rules for the TBox and ABox}
The static core extends standard description-logic reasoning by context labels.

\begin{definition}[Static rules]
The core static rules include the following schemata:
\[
\frac{\TBox \vdash C\sqsubseteq D \qquad \State_U \vdash a:C@U}
     {\State_U \vdash a:D@U}
\; (\mathrm{T\text{-}Sub})
\]
\[
\frac{\State_U \vdash a:C@U \qquad \State_U \vdash a:D@U}
     {\State_U \vdash a:(C\sqcap D)@U}
\; (\sqcap\mathrm{I})
\]
\[
\frac{\State_U \vdash a:(C\sqcap D)@U}
     {\State_U \vdash a:C@U}
\; (\sqcap\mathrm{E}_1)
\qquad
\frac{\State_U \vdash a:(C\sqcap D)@U}
     {\State_U \vdash a:D@U}
\; (\sqcap\mathrm{E}_2)
\]
\[
\frac{\State_U \vdash (a,b):r@U \qquad \State_U \vdash b:C@U}
     {\State_U \vdash a:\exists r.C@U}
\; (\exists\mathrm{I})
\]
\[
\frac{a:C@U\in \ABox_U}{\State_U \vdash a:C@U}
\; (\mathrm{A\text{-}Ax})
\qquad
\frac{(a,b):r@U\in \ABox_U}{\State_U \vdash (a,b):r@U}
\; (\mathrm{R\text{-}Ax})
\]
\end{definition}

\begin{remark}
These rules are intentionally modest. Their role is to exhibit a static deductive layer compatible with the categorical semantics of \cref{sec:ta}, not to provide a maximal calculus for all description-logic constructors.
\end{remark}

\subsection{Guard judgments as a distinct layer}\label{sec:guards}
The distinction between static description and procedural control is essential in TAPO-DL. In particular, the branching constructs of the PBox are not ordinary formulas of description logic. They are governed by a separate guard judgment at the metalevel. Formally, this judgment belongs neither to the TBox/ABox object language nor to the oracle language. In concrete interpretations it may be realized by a system, by a user, or by an oracle-mediated interaction process.

\begin{definition}[Guard language and guard profile]
Fix a context $U$. Let $G_U$ be a designated set of \emph{basic guard atoms}. In the present paper these may be chosen from localized assertions such as $a:C@U$ or $(a,b):r@U$ whenever they are intended to serve as branching conditions. The guard language $\mathcal G_U$ is generated by
\[
\gamma ::= \alpha \mid \top \mid \bot \mid (\gamma\wedge\gamma) \mid (\gamma\vee\gamma) \mid \neg \gamma,
\qquad \alpha\in G_U.
\]
A \emph{guard profile} at $U$ is a total assignment
\[
\chi_U:G_U\to\{\mathbf t,\mathbf f\}.
\]
The corresponding guard judgment
\[
(\State_U,\PBox_U)\Downarrow_g \gamma : b
\qquad (b\in\{\mathbf t,\mathbf f\})
\]
is the least relation determined by $\chi_U$ together with the usual Boolean clauses:
\[
(\State_U,\PBox_U)\Downarrow_g \alpha : b \iff \chi_U(\alpha)=b,
\]
\[
(\State_U,\PBox_U)\Downarrow_g \top : \mathbf t,
\qquad
(\State_U,\PBox_U)\Downarrow_g \bot : \mathbf f,
\]
\[
(\State_U,\PBox_U)\Downarrow_g (\gamma_1\wedge\gamma_2):\mathbf t
\iff
(\State_U,\PBox_U)\Downarrow_g \gamma_1:\mathbf t
\text{ and }
(\State_U,\PBox_U)\Downarrow_g \gamma_2:\mathbf t,
\]
\[
(\State_U,\PBox_U)\Downarrow_g (\gamma_1\wedge\gamma_2):\mathbf f
\iff
(\State_U,\PBox_U)\Downarrow_g \gamma_1:\mathbf f
\text{ or }
(\State_U,\PBox_U)\Downarrow_g \gamma_2:\mathbf f,
\]
\[
(\State_U,\PBox_U)\Downarrow_g (\gamma_1\vee\gamma_2):\mathbf t
\iff
(\State_U,\PBox_U)\Downarrow_g \gamma_1:\mathbf t
\text{ or }
(\State_U,\PBox_U)\Downarrow_g \gamma_2:\mathbf t,
\]
\[
(\State_U,\PBox_U)\Downarrow_g (\gamma_1\vee\gamma_2):\mathbf f
\iff
(\State_U,\PBox_U)\Downarrow_g \gamma_1:\mathbf f
\text{ and }
(\State_U,\PBox_U)\Downarrow_g \gamma_2:\mathbf f,
\]
\[
(\State_U,\PBox_U)\Downarrow_g \neg\gamma:\mathbf t
\iff
(\State_U,\PBox_U)\Downarrow_g \gamma:\mathbf f,
\qquad
(\State_U,\PBox_U)\Downarrow_g \neg\gamma:\mathbf f
\iff
(\State_U,\PBox_U)\Downarrow_g \gamma:\mathbf t.
\]
\end{definition}

\begin{remark}
The use of guard profiles avoids raw non-derivability conditions such as $\nvdash$ inside the procedural rules. This is important both proof-theoretically and conceptually. Proof-theoretically, it keeps branching on the side of positive judgments. Conceptually, it makes explicit that the decision to take a branch belongs to a procedural control layer and need not be identified with ordinary description-logic entailment.
\end{remark}

\subsection{Transition rules for the PBox}
The procedural layer is naturally formulated by transition judgments rather than ordinary formula entailment. Its branching rules are defined relative to the guard judgment introduced above.

\begin{definition}[Core PBox transition rules]
For programs in the language $\mathbf P_0$, the following transition rules are admitted:
\[
\frac{}{(\State_U,\PBox_U)\vdash \mathbf{skip}:\State_U\rightsquigarrow \State_U}
\; (\mathrm{Skip})
\]
\[
\frac{}{(\State_U,\PBox_U)\vdash \mathbf{add}\;\beta:\State_U\rightsquigarrow (\TBox,\ABox_U\cup\{\beta\})}
\; (\mathrm{Add})
\]
\[
\frac{}{(\State_U,\PBox_U)\vdash \mathbf{del}\;\beta:\State_U\rightsquigarrow (\TBox,\ABox_U\setminus\{\beta\})}
\; (\mathrm{Del})
\]
\[
\frac{(\State_U,\PBox_U)\vdash P:\State_U\rightsquigarrow \State_U' \qquad
      (\State_U',\PBox_U)\vdash Q:\State_U'\rightsquigarrow \State_U''}
     {(\State_U,\PBox_U)\vdash P;Q:\State_U\rightsquigarrow \State_U''}
\; (\mathrm{Seq})
\]
\[
\frac{(\State_U,\PBox_U)\Downarrow_g \gamma : \mathbf t \qquad (\State_U,\PBox_U)\vdash P:\State_U\rightsquigarrow \State_U'}
     {(\State_U,\PBox_U)\vdash \mathbf{if}\;\gamma\;\mathbf{then}\;P\;\mathbf{else}\;Q:\State_U\rightsquigarrow \State_U'}
\; (\mathrm{If\text{-}T})
\]
\[
\frac{(\State_U,\PBox_U)\Downarrow_g \gamma : \mathbf f \qquad (\State_U,\PBox_U)\vdash Q:\State_U\rightsquigarrow \State_U'}
     {(\State_U,\PBox_U)\vdash \mathbf{if}\;\gamma\;\mathbf{then}\;P\;\mathbf{else}\;Q:\State_U\rightsquigarrow \State_U'}
\; (\mathrm{If\text{-}F})
\]
\[
\frac{(\State_U,\PBox_U)\Downarrow_g \gamma : \mathbf f}
     {(\State_U,\PBox_U)\vdash \mathbf{while}\;\gamma\;\mathbf{do}\;P:
      \State_U\rightsquigarrow \State_U}
\; (\mathrm{While\text{-}F})
\]
\[
\frac{\begin{array}{c}
(\State_U,\PBox_U)\Downarrow_g \gamma : \mathbf t \\
(\State_U,\PBox_U)\vdash P:\State_U\rightsquigarrow \State_U' \\
(\State_U',\PBox_U)\vdash \mathbf{while}\;\gamma\;\mathbf{do}\;P:\State_U'\rightsquigarrow \State_U''
\end{array}}
{(\State_U,\PBox_U)\vdash \mathbf{while}\;\gamma\;\mathbf{do}\;P:\State_U\rightsquigarrow \State_U''}
\; (\mathrm{While\text{-}T})
\]
\end{definition}

\begin{remark}
The PBox rules express information behavior as derivable state change. In particular, a hesitation condition can lead not directly to a final assertion, but to a procedural branch, to repeated reformulation, or to invocation of the oracle layer. The guard judgment keeps explicit that such branching belongs to a distinct procedural control layer rather than to the descriptive core alone.
\end{remark}

\subsection{Validated import rules for the OBox}
The strengthened OBox from \cref{sec:obox} permits a proof-theoretic reading of trust-sensitive external import.

\begin{definition}[Core OBox rules]
Let $V$ be a validation policy. The core oracle rules are:
\[
\frac{q\in Q_U \qquad \operatorname{ans}_U(q)=r}
     {(\State_U,\OBox_U)\vdash q \leadsto r}
\; (\mathrm{Query})
\]
\[
\frac{(\State_U,\OBox_U)\vdash q \leadsto r \qquad
      \operatorname{tr}_U(r)\succeq \ell_U \qquad
      \operatorname{val}_U(r,S)=\mathsf{accept}}
     {(\State_U,\OBox_U)\vdash q \xRightarrow{V} (\TBox,\ABox_U\cup \operatorname{vimp}_U(r))}
\; (\mathrm{Oracle\text{-}Accept})
\]
\[
\frac{(\State_U,\OBox_U)\vdash q \leadsto r \qquad
      \bigl(\operatorname{tr}_U(r)\not\succeq \ell_U \text{ or } \operatorname{val}_U(r,S)\in\{\mathsf{reject},\mathsf{defer}\}\bigr)}
     {(\State_U,\OBox_U)\vdash q \xRightarrow{V} \State_U}
\; (\mathrm{Oracle\text{-}Hold})
\]
for finite certificate sets $S\subseteq \operatorname{cert}_U(r)$.
\end{definition}

\begin{remark}
The rule $(\mathrm{Oracle\text{-}Hold})$ is important for information-behavior analysis. An external review, recommendation, or report may be consulted without yet being imported; this corresponds formally to a derivable transition that leaves the current state unchanged.
\end{remark}

\subsection{Procedural invocation of the oracle layer}
The TAPO setting becomes distinctive when PBox and OBox judgments are allowed to interact.

\begin{definition}[Consultation rule]
Assume that a procedural state records unresolved hesitation by an assertion of the form
\[
\State_U \vdash x:\mathsf{ReviewConsultationNeeded}@U.
\]
Then an oracle consultation may be triggered by the rule
\[
\frac{\State_U \vdash x:\mathsf{ReviewConsultationNeeded}@U \qquad
      (\State_U,\OBox_U)\vdash q_x \xRightarrow{V} \State_U'}
     {(\State_U,\PBox_U,\OBox_U)\vdash \mathbf{consult}(q_x):\State_U\rightsquigarrow \State_U'}
\; (\mathrm{Consult})
\]
provided $q_x$ is an admissible query associated with the hesitation state of $x$.
\end{definition}

\begin{remark}
Rule $(\mathrm{Consult})$ formalizes exactly the sort of situation illustrated by the curry-restaurant example in Example~\ref{ex:curry2}: uncertainty does not directly resolve into acceptance or rejection, but first triggers an external review query, whose validated result then re-enters the procedural flow.
\end{remark}

\section{TAPO-structured knowledge objects}\label{sec:tapo-objects}
We now turn to the categorical semantics of the refined framework. The first step is to isolate the localized TAPO-objects on which the semantic constructions operate.

In order to keep the categorical language close to the original TAPO-DL notation, we separate the static knowledge state
\[
\State_U=(\TBox,\ABox_U)
\]
from the procedural and oracle layers attached to it.

\begin{definition}[Localized knowledge state]
For a fixed context $U\in\Ctx$, a \emph{localized knowledge state} is a pair
\[
\State_U=(\TBox,\ABox_U),
\]
where $\TBox$ is a TBox and $\ABox_U$ is a contextual ABox over $U$.
\end{definition}

\begin{definition}[TAPO-object]
A \emph{TAPO-object over $U$} is a triple
\[
X_U=(\State_U,\PBox_U,\OBox_U)
\]
consisting of a localized knowledge state $\State_U=(\TBox,\ABox_U)$, a procedural layer $\PBox_U$ over $\State_U$, and an oracle layer $\OBox_U$ over $\State_U$.
\end{definition}

\begin{definition}[Morphisms of TAPO-objects]
A morphism
\[
f:X_U\to Y_U
\]
of TAPO-objects over the same context is a quadruple
\[
f=(f_T,f_A,f_P,f_O)
\]
with the following properties:
\begin{enumerate}
    \item $f_T$ preserves valid terminological inclusions in the TBox;
    \item $f_A$ carries contextual assertions $a:C@U$ and $(a,b):r@U$ in $\ABox_U$ to corresponding assertions in the target ABox;
    \item $f_P$ transports admissible PBox programs and preserves their evaluations whenever defined;
    \item $f_O$ transports admissible oracle interfaces and preserves validated oracle imports whenever defined.
\end{enumerate}
\end{definition}

\begin{definition}[Execution-stable morphism]
A morphism $f:X_U\to Y_U$ is \emph{execution-stable} if the following hold.
\begin{enumerate}
    \item For every program $P\in \PBox_U$, the denotation of its transport $f_P(P)\in \PBox'_U$ satisfies
    \[
    f_A\circ \Sem{P}_U \simeq \Sem{f_P(P)}_U\circ f_A
    \]
    wherever both sides are defined.
    \item For every oracle interface $\omega\in \OBox_U$, the validated import transition of its transport $f_O(\omega)$ satisfies the analogous commutation relation with $f_A$.
\end{enumerate}
\end{definition}

\begin{proposition}[Category of TAPO-objects]
For each fixed context $U$, TAPO-objects over $U$ and execution-stable morphisms between them form a category.
\end{proposition}

\begin{proof}
Identity morphisms are execution-stable by definition. If $f:X_U\to Y_U$ and $g:Y_U\to Z_U$ are execution-stable, then preservation of TBox data, ABox assertions, program evaluations, and validated oracle imports is inherited by composition. Hence $g\circ f$ is again execution-stable.
\end{proof}

\begin{remark}
This formulation is closer to the original TAPO-DL presentation than the earlier quadruple-style description, because the PBox and OBox are now visibly attached to the localized state $\State_U=(\TBox,\ABox_U)$ on which they operate.
\end{remark}

\section{The TBox and ABox as categorical data}\label{sec:ta}
The TBox and ABox form the static core of the theory. This is the part most closely aligned with existing category semantics for ordinary description logic \cite{LeDuc2021,BrieulleLeDucVaillant2022,LeDucBrieulle2025}.

\subsection{The TBox}
A TBox is interpreted as a system of structural constraints. In a category-theoretic language, concept inclusions may be read as arrows, subobjects, or designated entailment morphisms, depending on the semantic level one chooses.

\begin{definition}[Categorical TBox semantics]
A \emph{categorical TBox semantics} assigns to each concept symbol $C$ an object $\Sem{C}$ and to each valid inclusion $C\sqsubseteq D$ in $\TBox$ a distinguished morphism
\[
\Sem{C}\longrightarrow \Sem{D}
\]
in a background category $\mathcal{E}$ with sufficient logical structure.
\end{definition}

\subsection{The ABox}
The ABox records individuals and role assertions. In a categorical semantics this may be encoded by generalized elements, sections, or points relative to context.

\begin{definition}[Contextual ABox semantics]
For each context $U\in\Ctx$, an assertion $a:C@U$ is interpreted by a generalized element
\[
1_U\longrightarrow \Sem{C}(U),
\]
and a role assertion $(a,b):r@U$ is interpreted by a generalized element of the contextual relation object corresponding to $r$ over $U$.
\end{definition}

\begin{proposition}[Restriction of assertions]
Let $V\leq U$ in $\Ctx$. If $a:C@U$ or $(a,b):r@U$ is represented over $U$, then restriction induces a corresponding representation over $V$. In particular, contextual assertions are functorial under refinement.
\end{proposition}

\begin{proof}
This is immediate from the contravariance of the presheaf associated with the relevant concept or relation object.
\end{proof}

\begin{proposition}[Preservation under TAPO-morphisms]
Let $f:X_U\to Y_U$ be a morphism of TAPO-objects. If $f$ respects the distinguished morphisms associated with terminological inclusions and the generalized elements representing contextual assertions, then every TBox entailment or ABox assertion valid in $X_U$ is carried to a corresponding valid entailment or assertion in $Y_U$.
\end{proposition}

\begin{proof}
The preservation of TBox data follows from functoriality on the distinguished entailment morphisms, and the preservation of ABox data follows from the image of the corresponding generalized elements under $f$.
\end{proof}

\section{The PBox as a minimal imperative layer}\label{sec:pbox}
We now move closer to the notation of the original TAPO-DL paper \cite{InoueTAPO}. The PBox is not a set of inference rules, but a programmable layer over the localized knowledge state
\[
\State_U=(\TBox,\ABox_U).
\]

\subsection{Guards}
\begin{definition}[Guards]
Guards over $\State_U=(\TBox,\ABox_U)$ are generated by
\[
\varphi ::= \top \mid \bot \mid \alpha \mid \neg\varphi \mid (\varphi \wedge \varphi),
\]
where atomic guards are of the form
\[
\alpha ::= a:C@U \mid (a,b):r@U \mid (C\sqsubseteq D).
\]
Their satisfaction relation is defined by
\begin{align*}
(\TBox,\ABox_U)\models a:C@U &\iff a:C@U\in \ABox_U,\\
(\TBox,\ABox_U)\models (a,b):r@U &\iff (a,b):r@U\in \ABox_U,\\
(\TBox,\ABox_U)\models (C\sqsubseteq D) &\iff \TBox\vdash_{\mathrm{DL}} C\sqsubseteq D,
\end{align*}
with Boolean connectives interpreted classically.
\end{definition}

\subsection{Program syntax}
\begin{definition}[The minimal imperative language $\mathbf P_0$]
The language $\mathbf P_0$ of localized PBox programs is generated by
\[
P ::= \mathbf{skip}
 \mid \mathbf{add}\;\beta
 \mid \mathbf{del}\;\beta
 \mid P;P
 \mid \mathbf{if}\;\gamma\;\mathbf{then}\;P\;\mathbf{else}\;P
 \mid \mathbf{while}\;\gamma\;\mathbf{do}\;P,
\]
where $\beta$ ranges over ABox assertions $a:C@U$ and $(a,b):r@U$, and $\gamma$ ranges over the guard language $\mathcal G_U$ of \cref{sec:guards}.
\end{definition}

\begin{remark}
This is exactly the style of programmable procedural layer emphasized in \cite{InoueTAPO}. The categorical step taken here is not to alter the syntax, but to interpret such programs as partial endomaps of localized knowledge states.
\end{remark}

\subsection{Operational semantics and denotations}
We write
\[
\langle P,\State_U\rangle \Downarrow \State'_U
\]
for the big-step evaluation relation, understood relative to the guard profile carried by the current procedural configuration at $U$. Thus the clauses for conditionals and while-loops consult the guard judgment of \cref{sec:guards}, while the remaining clauses are the usual ones for \textbf{skip}, \textbf{add}, \textbf{del}, and sequencing.

\begin{definition}[Program denotation]
For each program $P\in \mathbf P_0$, its denotation at context $U$ is the partial map
\[
\Sem{P}_U:\State_U\dashrightarrow \State_U
\]
given by
\[
\Sem{P}_U(\State_U)=\State'_U
\quad\Longleftrightarrow\quad
\langle P,\State_U\rangle \Downarrow \State'_U.
\]
\end{definition}

\begin{definition}[PBox semantics]
A \emph{PBox over $U$} is a chosen class
\[
\PBox_U\subseteq \mathbf P_0
\]
of admissible programs whose denotations are closed under the semantic operations that are intended to be available at context $U$.
\end{definition}

\begin{proposition}[Sequential compositionality]
If $P_1,P_2\in \PBox_U$ and both denotations are defined on a localized state $\State_U$, then
\[
\Sem{P_2;P_1}_U = \Sem{P_2}_U\circ \Sem{P_1}_U
\]
wherever the composite is defined. In particular, the PBox carries an intrinsic compositional structure.
\end{proposition}

\begin{proof}
This is the standard big-step semantics for sequencing.
\end{proof}

\begin{proposition}[Closure under conditionals]
Assume $\PBox_U$ is closed under the program-forming operation
\[
(P_1,P_2,\gamma)\longmapsto \mathbf{if}\;\gamma\;\mathbf{then}\;P_1\;\mathbf{else}\;P_2.
\]
Then every such conditional built from members of $\PBox_U$ again belongs to $\PBox_U$.
\end{proposition}

\begin{proof}
This is immediate from the closure assumption.
\end{proof}

\begin{remark}
The preceding proposition is the direct categorical analogue of the original TAPO-DL idea that conditionals are internal program constructors, not meta-level instructions.
\end{remark}

\begin{definition}[Iterative partiality]
For a guard $\gamma$ and a program $P\in \PBox_U$, the loop
\[
\mathbf{while}\;\gamma\;\mathbf{do}\;P
\]
induces a partial endomap
\[
\Sem{\mathbf{while}\;\varphi\;\mathbf{do}\;P}_U:\State_U\dashrightarrow \State_U
\]
which is defined exactly on those states for which the iteration terminates.
\end{definition}

\begin{remark}
Thus the PBox naturally accommodates partiality. This matches the original proof-theoretic intuition of TAPO-DL, where \textbf{while}-constructs need not terminate \cite{InoueTAPO}.
\end{remark}

\begin{proposition}[Restriction compatibility of procedures]
Suppose $V\leq U$ and the procedural layer is stable under contextual restriction. Then for every $P\in \PBox_U$ there exists a restricted program $P|_V\in \PBox_V$ such that
\[
\rho_{UV}\circ \Sem{P}_U \simeq \Sem{P|_V}_V\circ \rho_{UV}
\]
wherever both sides are defined.
\end{proposition}

\begin{proof}
This is the procedural analogue of functoriality for contextual assertions. Stability under restriction provides the restricted program, and the commutation statement is precisely the assumed contextual compatibility.
\end{proof}

\begin{remark}
It is worth noting, as a nearby line of motivation rather than a direct semantic precursor, that recent work in NLP-assisted formal modeling has attempted to extract from natural-language specifications intermediate structures built from predicates, functions, and pre-/post-condition patterns \cite{GhoshEtAl2022SpecNFS}. From the present viewpoint, such representations are relevant mainly as a syntactic front end for procedural information. The role of the PBox in TAPO-description logic is different: it is meant to internalize that procedural layer as mathematical structure inside the logic itself.
\end{remark}

\section{Sheaf-theoretic refinement}\label{sec:sheaf}
The presheaf picture suggests a genuine sheaf-theoretic strengthening.

\begin{definition}[Sheaf of TAPO-states]
A presheaf of TAPO-states is a \emph{sheaf} if locally compatible families of TAPO-objects admit unique gluing.
\end{definition}

\begin{theorem}[Schematic local-to-global principle]
Assume the TAPO-state assignment forms a sheaf on $\Ctx$, and assume that both procedural denotations and strengthened oracle frames are restriction-compatible. Then compatible local TAPO-updates glue to a global TAPO-update.
\end{theorem}

\begin{proof}
At this stage the statement is schematic. The intended proof combines the sheaf condition for underlying local states with compatibility of PBox and OBox data under restriction. Once these data are formulated as morphisms in an appropriate fibred or internal category, the gluing follows from the corresponding descent condition.
\end{proof}

\begin{remark}
This is the point where a later topos-theoretic development becomes natural. Local knowledge, local procedures, and local oracle responses can then be studied internally, while global coherence is governed by descent.
\end{remark}

\begin{remark}
In particular, a trust- and certification-sensitive OBox suggests that descent should later be formulated not only for bare assertions, but also for validated provenance data. This would make the local-to-global passage sensitive to how imported information was certified, not merely to what was imported.
\end{remark}

\section{Soundness and the completeness problem}\label{sec:soundcomplete}
The proof theory is intended as a companion to the semantic structures already introduced.

\begin{theorem}[Schematic soundness]
Each of the preceding rules is sound with respect to the semantic skeleton of \cref{sec:ta,sec:pbox,sec:obox}. More precisely:
\begin{enumerate}
    \item if $\TBox\vdash C\sqsubseteq D$ or $\State_U\vdash \alpha$ is derivable by the static rules, then the corresponding categorical interpretation is valid in the sense of \cref{sec:ta};
    \item if $(\State_U,\PBox_U)\vdash P:\State_U\rightsquigarrow \State_U'$ is derivable, then $\Sem{P}_U(\State_U)=\State_U'$ whenever the denotation is defined;
    \item if $(\State_U,\OBox_U)\vdash q\xRightarrow{V}\State_U'$ is derivable, then $\State_U'$ coincides with the state produced by the validated oracle transition associated with the strengthened OBox.
\end{enumerate}
The mixed consultation rule is also sound: if
\[
(\State_U,\PBox_U,\OBox_U)\vdash \mathbf{consult}(q_x):\State_U\rightsquigarrow \State_U',
\]
then the target state is exactly the one obtained by the validated oracle update triggered from the hesitation state of $x$.
\end{theorem}

\begin{proof}
We argue by induction on the last rule used in the derivation.

\smallskip
\noindent\emph{(1) Static rules.}
The axiom rules $(\mathrm{A\text{-}Ax})$ and $(\mathrm{R\text{-}Ax})$ are immediate from the contextual ABox semantics of \cref{sec:ta}: an asserted judgment belongs to $\ABox_U$, hence is represented by the corresponding generalized element over $U$.

For $(\mathrm{T\text{-}Sub})$, suppose that $\TBox\vdash C\sqsubseteq D$ and $\State_U\vdash a:C@U$ are both derivable. By the induction hypothesis, the premise $a:C@U$ is interpreted by a generalized element
\[
1_U\longrightarrow \Sem{C}(U),
\]
and the inclusion $C\sqsubseteq D$ is interpreted by a distinguished morphism
\[
\Sem{C}\longrightarrow \Sem{D}
\]
in the categorical TBox semantics. Composing these arrows yields a generalized element
\[
1_U\longrightarrow \Sem{D}(U),
\]
which is exactly the interpretation of $a:D@U$.

For $(\sqcap\mathrm{I})$, if $a:C@U$ and $a:D@U$ are both derivable, then by induction they are represented by generalized elements into $\Sem{C}(U)$ and $\Sem{D}(U)$. In the intended semantics, conjunction is interpreted by the categorical meet or product-like conjunction object, so these two arrows induce an arrow into $\Sem{C\sqcap D}(U)$. This yields the soundness of conjunction introduction. The elimination rules $(\sqcap\mathrm{E}_1)$ and $(\sqcap\mathrm{E}_2)$ are then sound by composing with the corresponding projections.

For $(\exists\mathrm{I})$, assume that $(a,b):r@U$ and $b:C@U$ are derivable. By induction, we have a generalized element of the contextual relation object corresponding to $r$ together with a generalized element of $\Sem{C}(U)$ for $b$. These data determine a witness for the existential condition, hence a generalized element interpreting $a:\exists r.C@U$.

\smallskip
\noindent\emph{(2) Procedural rules.}
By \cref{sec:pbox}, the denotation $\Sem{P}_U$ is defined by the big-step evaluation relation
\[
\langle P,\State_U\rangle\Downarrow \State_U'.
\]
Thus each procedural rule is sound provided it agrees with the corresponding evaluation clause.

The rule $(\mathrm{Skip})$ is sound because the operational semantics of $\mathbf{skip}$ leaves the state unchanged, so $\Sem{\mathbf{skip}}_U$ is the identity on localized knowledge states. The rules $(\mathrm{Add})$ and $(\mathrm{Del})$ are sound because they coincide exactly with adjoining or removing the indicated ABox assertion from the current state.

For $(\mathrm{Seq})$, assume the premises derive
\[
(\State_U,\PBox_U)\vdash P:\State_U\rightsquigarrow \State_U'
\qquad\text{and}\qquad
(\State_U',\PBox_U)\vdash Q:\State_U'\rightsquigarrow \State_U''.
\]
By the induction hypothesis,
\[
\Sem{P}_U(\State_U)=\State_U'
\qquad\text{and}\qquad
\Sem{Q}_U(\State_U')=\State_U''.
\]
Hence the composite program satisfies
\[
\Sem{P;Q}_U(\State_U)=\State_U'',
\]
which is exactly the semantic content of the sequencing rule. This is the same compositionality recorded in \cref{sec:pbox}.

For $(\mathrm{If\text{-}T})$, if $(\State_U,\PBox_U)\Downarrow_g \gamma : \mathbf t$ is derivable, then the guard judgment selects the \emph{then}-branch. By induction on the premise for $P$, the resulting state is $\Sem{P}_U(\State_U)$, so the entire conditional has that same denotation. The case $(\mathrm{If\text{-}F})$ is analogous, except that the guard judgment derives $\mathbf f$ and the \emph{else}-branch is selected.

For $(\mathrm{While\text{-}F})$, if $(\State_U,\PBox_U)\Downarrow_g \gamma : \mathbf f$, then the loop terminates immediately and the target state remains $\State_U$, exactly as stated by the rule.

For $(\mathrm{While\text{-}T})$, assume that $(\State_U,\PBox_U)\Downarrow_g \gamma : \mathbf t$, that
\[
(\State_U,\PBox_U)\vdash P:\State_U\rightsquigarrow \State_U',
\]
and that
\[
(\State_U',\PBox_U)\vdash \mathbf{while}\;\gamma\;\mathbf{do}\;P:\State_U'\rightsquigarrow \State_U''.
\]
By the induction hypothesis, the first procedural premise is interpreted by $\Sem{P}_U(\State_U)=\State_U'$, and the second by the denotation of the loop started from $\State_U'$. Hence the entire loop is interpreted by first executing one step of $P$ and then iterating again. This is exactly the operational reading of the loop under guard-driven big-step evaluation.

\smallskip
\noindent\emph{(3) Oracle rules.}
For $(\mathrm{Query})$, the conclusion merely records the value of the partial answer map $\operatorname{ans}_U$. Hence any derivable query-response judgment agrees with the strengthened oracle frame by definition.

For $(\mathrm{Oracle\text{-}Accept})$, assume that
\[
(\State_U,\OBox_U)\vdash q\leadsto r,
\qquad
\operatorname{tr}_U(r)\succeq \ell_U,
\qquad
\operatorname{val}_U(r,S)=\mathsf{accept}.
\]
Then $r$ is a validated response in the sense of \cref{sec:obox}. Therefore $\operatorname{vimp}_U(r)$ is defined and agrees with $\operatorname{imp}_U(r)$ on that response. The conclusion of the rule is precisely the state obtained by adjoining the validated imported assertions to $\ABox_U$. By the soundness requirement imposed on strengthened OBoxes, those imported assertions remain compatible with $\TBox$, so the target is again a legitimate localized state. Hence the rule agrees exactly with the certified oracle transition of \cref{sec:obox}.

For $(\mathrm{Oracle\text{-}Hold})$, either the trust threshold is not met or the validation policy rejects or defers the response. In all such cases the validated import map is undefined, so no new ABox information may be integrated. The semantic effect is therefore the identity transition on $\State_U$, exactly as stated by the rule.

\smallskip
\noindent\emph{(4) The consultation rule.}
Assume that the last rule is $(\mathrm{Consult})$. Then we have a derivation of
\[
\State_U\vdash x:\mathsf{ReviewConsultationNeeded}@U
\]
and a derivation of
\[
(\State_U,\OBox_U)\vdash q_x\xRightarrow{V}\State_U'.
\]
By the static soundness already proved, the hesitation assertion is semantically valid in the current state. By oracle soundness, the second premise identifies $\State_U'$ with the output of the validated oracle update associated with $q_x$. The consultation rule does not add any independent semantic transformation beyond triggering that validated update from an admissible hesitation state. Therefore the conclusion
\[
(\State_U,\PBox_U,\OBox_U)\vdash \mathbf{consult}(q_x):\State_U\rightsquigarrow \State_U'
\]
is sound.

This completes the induction.
\end{proof}

\begin{remark}
A full completeness theorem is deliberately postponed. The present place in the exposition is intentional: once both the inference layer and the categorical semantics have been displayed, the completeness problem can be stated in its proper form. In particular, any later completeness treatment should separate the static guard layer from the dynamic transition layer and explain how guard profiles are fixed or generated. Because the present paper is centered on categorical semantics and information behavior rather than on a maximal proof calculus, the main point here is to show that TAPO-description logic already supports a coherent core of derivable static, guard-based procedural, and oracle-sensitive judgments.
\end{remark}

\section{Information-behavior examples and discussion}\label{sec:examples}
The preceding framework is intended for the analysis of information behavior rather than for static ontology management alone. We therefore gather in one place a sequence of examples that illustrate how the refined TAPO-DL architecture treats procedural hesitation, external consultation, validation-sensitive import, and final action.

\subsection{Simple search behavior}
The first group of examples models elementary information-seeking behavior. The point is that even a simple search episode already involves more than a static TBox/ABox-description: the agent may need procedural reformulation, iterative stabilization, and validated import of externally retrieved material. In this respect the examples are also close in spirit to Chang's account of browsing as a multidimensional activity rather than a purely query-driven terminal act \cite{ChangRice1993,RiceMcCreadieChang2001}.

\begin{example}[Information-seeking episode as a PBox program]
Let $U$ be a context representing the current information need of an agent. Assume that the localized ABox contains assertions of the form
\[
q:\mathsf{Query}@U,
\qquad
q:\mathsf{Underspecified}@U,
\qquad
s:\mathsf{SearchSession}@U.
\]
Consider the guard atom
\[
\gamma_{\mathrm{ref}} := q:\mathsf{Underspecified}@U
\]
viewed as an element of the designated basic guard set $G_U$, and the program
\[
P_{\mathrm{seek}} :=
\mathbf{if}\;\gamma_{\mathrm{ref}}\;\mathbf{then}\;
   \mathbf{add}\;q:\mathsf{NeedsRefinement}@U
\mathbf{else}\;
   \mathbf{add}\;q:\mathsf{ReadyForRetrieval}@U.
\]
Its denotation distinguishes two information-behavioral situations. If the guard judgment derives
\[
(\State_U,\PBox_U)\Downarrow_g \gamma_{\mathrm{ref}} : \mathbf t,
\]
then the state is updated so as to record the need for reformulation; if it derives $\mathbf f$, then the state is updated so as to permit retrieval. In this way, the PBox does not merely describe a search action externally, but internalizes the agent's procedural decision rule inside the localized TAPO-state.
\end{example}

\begin{example}[Iterative stabilization of a search process]
Still over a context $U$, suppose that the ABox may contain assertions
\[
q:\mathsf{NeedsRefinement}@U,
\qquad
q:\mathsf{ReadyForRetrieval}@U,
\qquad
q:\mathsf{StableResultSet}@U.
\]
Let
\[
P_{\mathrm{revise}} :=
\mathbf{del}\;q:\mathsf{NeedsRefinement}@U;
\mathbf{add}\;q:\mathsf{ReadyForRetrieval}@U
\]
and consider the loop
\[
\mathbf{while}\;\neg(q:\mathsf{StableResultSet}@U)\;\mathbf{do}\;P_{\mathrm{revise}}.
\]
Here the negated guard is evaluated by the guard judgment of \cref{sec:guards}. As long as
\[
(\State_U,\PBox_U)\Downarrow_g \neg(q:\mathsf{StableResultSet}@U) : \mathbf t,
\]
the revision step is executed and the loop continues; once the same judgment yields $\mathbf f$, the loop terminates via $(\mathrm{While\text{-}F})$. This gives a simple model of iterative information behavior: the agent repeatedly revises the search state until a stable result set is reached. The loop may remain partial, reflecting the familiar fact that information seeking need not converge in finitely many steps.
\end{example}

\begin{example}[Validated external retrieval in information behavior]
Let $U$ encode an active literature-search context. Suppose that the current localized state contains
\[
q:\mathsf{ReadyForRetrieval}@U,
\qquad
t:\mathsf{Topic}@U,
\qquad
r:\mathsf{ResultCandidate}@U.
\]
An admissible query in $Q_U$ asks an external bibliographic service for references relevant to $t$. A response $\rho\in R_U$ may contain candidate assertions such as
\[
r:\mathsf{RelevantTo}(t)@U,
\qquad
r:\mathsf{PeerReviewed}@U,
\qquad
r:\mathsf{RecentSource}@U.
\]
The strengthened OBox treats these assertions as importable only after validation. For example, the validation policy may require both a sufficient trust level and a certificate that the source metadata is complete. Only then are the imported assertions adjoined to the ABox. Thus the OBox models a basic informational distinction between \emph{retrieved} material and \emph{accepted} material.
\end{example}

\begin{example}[Trust-sensitive oracle use for credibility assessment]
Let $U$ represent a fact-checking context in which an agent is evaluating a claim $c$. Assume that responses may carry trust labels such as
\[
\mathsf{low} \prec \mathsf{medium} \prec \mathsf{high}
\]
and certificates recording provenance, timestamp, or source agreement. A validation policy can then require, for instance, that a response supporting $c$ be accepted only if its trust level is at least $\mathsf{medium}$ and at least one provenance certificate is present. In that case, the OBox does not merely import information from outside; it records a mathematically explicit policy by which the agent discriminates between usable and non-usable external information. This is precisely the kind of informational evaluation that motivated the passage from TBox/ABox-structured to TAPO-structured description logic.
\end{example}

\subsection{Review-sensitive ordering behavior in a curry restaurant}
The next examples show the same architecture in a more everyday setting. Menu labels produce impressions, those impressions trigger procedural branching, and online reviews function as oracle-type external input whose import is governed by trust and validation.

The interaction between the procedural and oracle layers becomes clearer when one considers an everyday information-behavioral situation in which menu perception, hesitation, review consultation, and final action are intertwined. See \cref{fig:curry-flow} for a schematic overview of this passage from menu data to final ordering behavior.

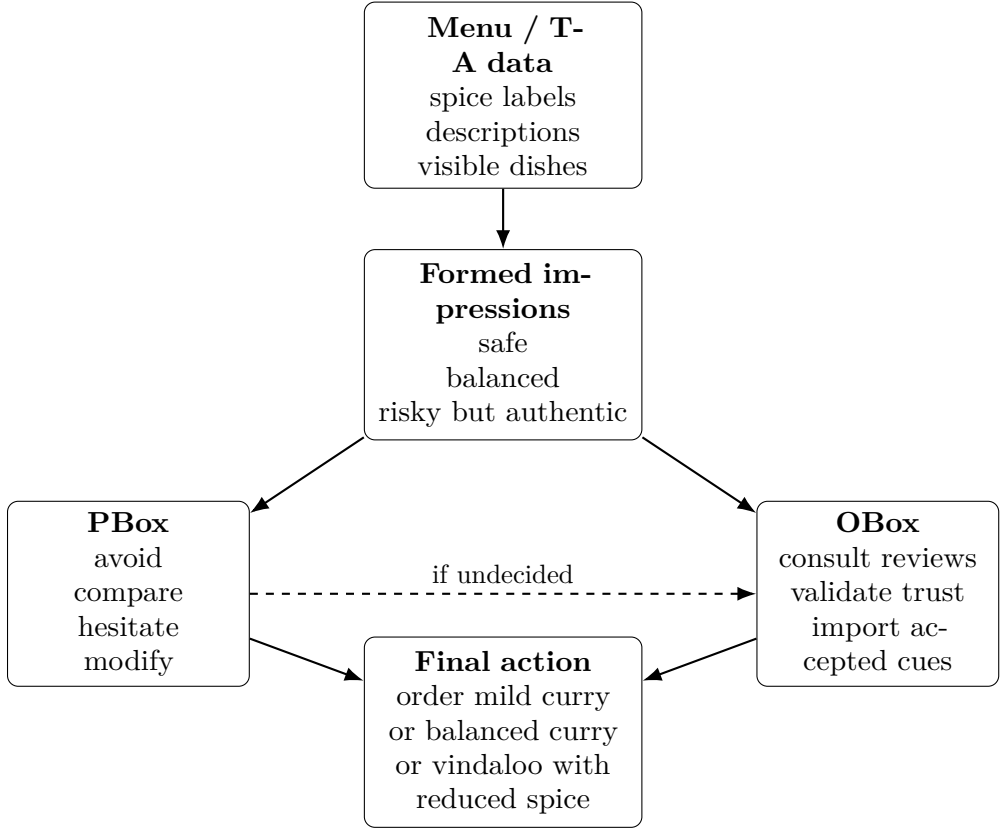
\begin{figure}[H]
\centering
\begin{tikzpicture}[
node distance=8mm and 9mm,
box/.style={draw, rounded corners, align=center, inner sep=5pt, text width=3.3cm},
smallbox/.style={draw, rounded corners, align=center, inner sep=4pt, text width=2.9cm},
arrow/.style={-{Latex[length=2.5mm]}, thick}
]
\node[box] (menu) {\textbf{Menu / T-A data}\\
spice labels\\
descriptions\\
visible dishes};

\node[box, below=of menu] (impr) {\textbf{Formed impressions}\\
safe\\
balanced\\
risky but authentic};

\node[smallbox, below left=of impr, xshift=-6mm] (pbox) {\textbf{PBox}\\
avoid\\compare\\hesitate\\modify};

\node[smallbox, below right=of impr, xshift=6mm] (obox) {\textbf{OBox}\\consult reviews\\validate trust\\import accepted cues};

\node[box, below=26mm of impr] (final) {\textbf{Final action}\\
order mild curry\\or balanced curry\\or vindaloo with reduced spice};

\draw[arrow] (menu) -- (impr);
\draw[arrow] (impr) -- (pbox);
\draw[arrow] (impr) -- (obox);
\draw[arrow] (pbox) -- (final);
\draw[arrow] (obox) -- (final);
\draw[arrow, dashed] (pbox.east) -- node[above, sloped, font=\small] {if undecided} (obox.west);
\end{tikzpicture}
\caption{Information-behavior flow in the curry restaurant example. Menu data first generate localized impressions; these are then processed procedurally in the PBox, possibly supplemented by validated review input through the OBox, before a final ordering action is reached.}
\label{fig:curry-flow}
\end{figure}

\begin{example}[Curry menu perception, review consultation, and ordering behavior]\label{ex:curry}
Let $U$ be a dining context in which two customers $u$ and $v$ face the same curry menu. Suppose the localized ABox contains assertions of the form
\[
u:\mathsf{Customer}@U,\qquad
v:\mathsf{Customer}@U,
\]
\[
c_1:\mathsf{Curry}@U,\qquad
c_2:\mathsf{Curry}@U,\qquad
c_3:\mathsf{Curry}@U,
\]
together with
\[
c_1:\mathsf{LowSpice}@U,\qquad
c_2:\mathsf{MediumSpice}@U,\qquad
c_3:\mathsf{HighSpice}@U.
\]
Assume moreover that the menu descriptions induce the assertions
\[
\begin{aligned}
&c_1:\mathsf{GentleMenuImpression}@U,\qquad
c_2:\mathsf{BalancedMenuImpression}@U,\\
&c_3:\mathsf{AuthenticMenuImpression}@U.
\end{aligned}
\]
and that the high-spice marking of $c_3$ also yields
\[
c_3:\mathsf{RiskyMenuImpression}@U.
\]
Thus the same dish $c_3$ is perceived under two competing aspects: it appears attractive as authentic, but risky as very hot. This tension is precisely the kind of localized informational state that cannot be expressed by the TBox/ABox core alone.
\end{example}

\begin{example}[PBox and OBox interaction in review-sensitive ordering]\label{ex:curry2}
Continue \cref{ex:curry}. Let the procedural layer contain a basic guard atom
\[
\gamma_{\mathrm{risk}}:=c_3:\mathsf{RiskyMenuImpression}@U
\]
and a composite guard
\[
\gamma_{\mathrm{undec}}:=u:\mathsf{UndecidedCustomer}@U
\vee
v:\mathsf{UndecidedCustomer}@U.
\]
For the customer $u$, who does not consult external reviews, consider the program
\[
\begin{aligned}
P_u :=\;&
\mathbf{if}\;\gamma_{\mathrm{risk}}\;\mathbf{then}\;
   \mathbf{add}\;c_3:\mathsf{AvoidCandidate}@U\\
&\mathbf{else}\;
   \mathbf{add}\;c_3:\mathsf{PreferredCandidate}@U;\
   \mathbf{add}\;c_1:\mathsf{SafeCandidate}@U.
\end{aligned}
\]
Its denotation models a cautious ordering strategy: if the guard judgment derives $(\State_U,\PBox_U)\Downarrow_g \gamma_{\mathrm{risk}}:\mathbf t$, then the strong-spice dish $c_3$ is excluded from the candidate set; otherwise the mild dish $c_1$ is promoted as the safe fallback.

For the customer $v$, by contrast, the procedural layer contains an intermediate step
\[
\begin{aligned}
P_v :=\;&
\mathbf{if}\;(\gamma_{\mathrm{risk}}\wedge \gamma_{\mathrm{undec}})\;\mathbf{then}\;
   \mathbf{add}\;v:\mathsf{ReviewConsultationNeeded}@U\\
&\mathbf{else}\;
   \mathbf{add}\;c_2:\mathsf{BalancedCandidate}@U.
\end{aligned}
\]
This means that the risk-impression attached to $c_3$ does not yet force rejection. Instead, when the guard judgment derives $(\State_U,\PBox_U)\Downarrow_g (\gamma_{\mathrm{risk}}\wedge \gamma_{\mathrm{undec}}):\mathbf t$, it triggers an information-seeking action that activates the oracle layer.

Now let $q_1\in Q_U$ ask whether the local meaning of ``spice level 3'' is milder than expected, and let $q_2\in Q_U$ ask whether the dish $c_3$ can be made less spicy on request. Suppose the oracle returns responses $r_1,r_2\in R_U$ such that
\[
\operatorname{imp}_U(r_1)=\{c_2:\mathsf{MilderThanExpected}@U\},
\]
\[
\operatorname{imp}_U(r_2)=\{c_3:\mathsf{AdjustableOnRequest}@U\}.
\]
Assume further that $r_1$ and $r_2$ both satisfy the relevant trust threshold and validation policy, so that they are accepted by the strengthened OBox and imported into the localized ABox.

After this validated oracle import, the procedural layer may continue with
\[
\begin{aligned}
P_v' :=\;&
\mathbf{if}\;c_3:\mathsf{AdjustableOnRequest}@U\;\mathbf{then}\;
   \mathbf{add}\;c_3:\mathsf{ControllableCandidate}@U\\
&\mathbf{else}\;
   \mathbf{add}\;c_2:\mathsf{BalancedCandidate}@U.
\end{aligned}
\]
Again the branching is governed by the guard judgment associated with the updated procedural configuration. Thus the ordering behavior of $v$ differs essentially from that of $u$. Customer $u$, who relies only on menu perception, may end with the mild order
\[
u:\mathsf{Orders}(c_1)@U,
\]
whereas customer $v$, after trusted review consultation, may reach either
\[
v:\mathsf{Orders}(c_2)@U
\]
as a balanced choice, or even
\[
v:\mathsf{Orders}(c_3)@U,\qquad
v:\mathsf{ReducedSpiceRequest}(c_3)@U.
\]
This example illustrates the informational reason for passing from TBox/ABox-structured to TAPO-structured description logic. The TBox/ABox layers record the menu and its explicit labels; the PBox records hesitation, candidate selection, and procedural branching; and the OBox records the import of external review information under a trust-sensitive validation policy. Online reviews therefore function here as a genuine oracle-like source: they do not merely add extra assertions, but alter how menu information is interpreted and how the final ordering action is procedurally determined.
\end{example}

\begin{remark}
The curry example is intentionally elementary, but it already displays the information-behavioral content of TAPO-description logic. A single menu label can produce competing impressions; those impressions can trigger procedural hesitation rather than immediate action; and external review material can then enter the state only after validation. In this way, the final behavior is governed not by static data alone, but by the interaction of description, procedure, and oracle-sensitive evaluation.
\end{remark}

\begin{remark}
Taken together, the search examples and the curry examples show why the passage from TBox/ABox-structured to TAPO-structured description logic is not merely decorative. The TBox/ABox core records explicit descriptions; the PBox records hesitation, branching, revision, and candidate management; and the OBox records the admission of external material under explicit validation conditions. This is exactly the level at which one can begin to analyze and evaluate information behavior rather than merely describe information objects.
\end{remark}

\section{Further directions}\label{sec:future}
This paper is intentionally skeletal. Its aim is to prepare a mathematically usable categorical language for TAPO-description logic while keeping clear distance from a mere literature review.

The next steps seem to include at least the following:
\begin{enumerate}
    \item a precise base category or topos in which concepts, individuals, programs, and strengthened oracle frames all live internally;
    \item a fuller proof theory, including completeness questions, normalization issues, and a more systematic comparison with the original TAPO-DL framework of \cite{InoueTAPO};
    \item a richer internal treatment of trust, certification, provenance, and validation in the oracle layer;
    \item richer families of examples showing how PBox and OBox data interact in concrete information-behavior scenarios, especially when procedural updates depend on validated external input and trust-sensitive review material.
\end{enumerate}

The present paper should therefore be read as a categorical first step beyond the original logical formulation \cite{InoueTAPO}, and as a bridge from ordinary categorical semantics of description logic to a broader TAPO-structured theory of information behavior.

\appendix

\section{TAPO-Description Logic as a Formal Analytic Layer for Browsing Theory}\label{app:chang-browsing}

The purpose of this appendix is to indicate that the present framework
is not only a formal system for information behavior in the abstract,
but may also serve as an analytic layer for established theories of browsing.
In particular, multidimensional approaches to browsing,
such as those associated with Chang and collaborators
\cite{ChangRice1993,RiceMcCreadieChang2001},
suggest that browsing should not be reduced to a weak or incomplete form
of direct searching.
Rather, browsing is shaped by context, motivation, cognitive stance,
available resources, and interface structure.
From the present viewpoint,
these aspects can be reorganized in a TAPO-structured manner
and thereby made accessible to formal comparison and rule-based analysis.

\subsection{Why Chang-style browsing is relevant here}

A central lesson of the browsing literature is that information behavior
often unfolds without a fully fixed query or a fully specified target state.
Users may begin with only a vague interest,
a tentative orientation,
or a partially formed need,
and the process may involve repeated inspection,
comparison, hesitation, redirection, and opportunistic selection.
Moreover, browsing behavior is influenced not only by the informational objects
themselves,
but also by the manner in which they are organized and displayed,
the interface affordances available to the user,
and the external evaluative cues encountered during the process
\cite{ChangRice1993,ChangAmazon2001}.

This perspective is highly compatible with the present development.
If one works only with a TBox/ABox description-logic layer,
one may represent concepts, instances, and static relational facts,
but it becomes difficult to express the procedural unfolding of browsing
or the selective use of external evaluative inputs.
The introduction of the PBox and OBox is therefore not merely
an optional extension:
it is precisely what allows one to treat browsing
as a dynamic and context-sensitive information behavior.

\subsection{A TAPO reinterpretation of multidimensional browsing}

Let $U$ be a local informational context,
for example a current interface state,
a shelf region,
a webpage category,
or a visible result cluster.
A TAPO-structured browsing state over $U$ is written
\[
X_U=(\State_U,\PBox_U,\OBox_U),\qquad \State_U=(\TBox,\ABox_U).
\]
Figure~\ref{fig:chang-tapo-browsing} summarizes the basic TAPO reinterpretation
of a Chang-style browsing episode.

\begin{figure}[H]
\centering
\begin{tikzpicture}[
  box/.style={draw, rounded corners, align=center, minimum height=9mm, inner sep=4pt, font=\small},
  line/.style={-Latex, thick}
]
\node[box, text width=9.6cm] (context) at (0,0) {\textbf{Local browsing context $U$}\\
interface state, shelf region, webpage category,\\
visible result cluster};

\node[box, text width=4.8cm] (ta) at (-3.3,-2.6) {\textbf{TBox/ABox level}\\
static organizational structure\\
and local instantiated facts\\
visible items, categories, links,\\
hesitation, candidate status};

\node[box, text width=4.0cm] (obox) at (3.3,-2.6) {\textbf{OBox}\\
reviews, rankings,\\
recommendations,\\
trust assignment,\\
certification,\\
validation policy};

\node[box, text width=9.6cm] (pbox) at (0,-6.0) {\textbf{PBox}\\
inspect, compare, follow link, save candidate, return, refine interest\\
procedural browsing behavior};

\node[box, text width=9.6cm] (outcome) at (0,-8.5) {\textbf{Browsing outcome}\\
stabilized candidate set, selected item, postponed decision, redirected exploration};

\draw[line] ([xshift=-16mm]context.south) -- (ta.north);
\draw[line] ([xshift=16mm]context.south) -- (obox.north);
\draw[line] (ta.south) -- ([xshift=-20mm]pbox.north);
\draw[line] (obox.south) -- ([xshift=20mm]pbox.north);
\draw[line] (pbox.south) -- (outcome.north);
\end{tikzpicture}
\caption{A TAPO reinterpretation of multidimensional browsing in the sense of Chang and collaborators. Static organizational structure is represented at the TBox/ABox level, browsing actions at the PBox level, and selectively incorporated external cues at the OBox level.}
\label{fig:chang-tapo-browsing}
\end{figure}
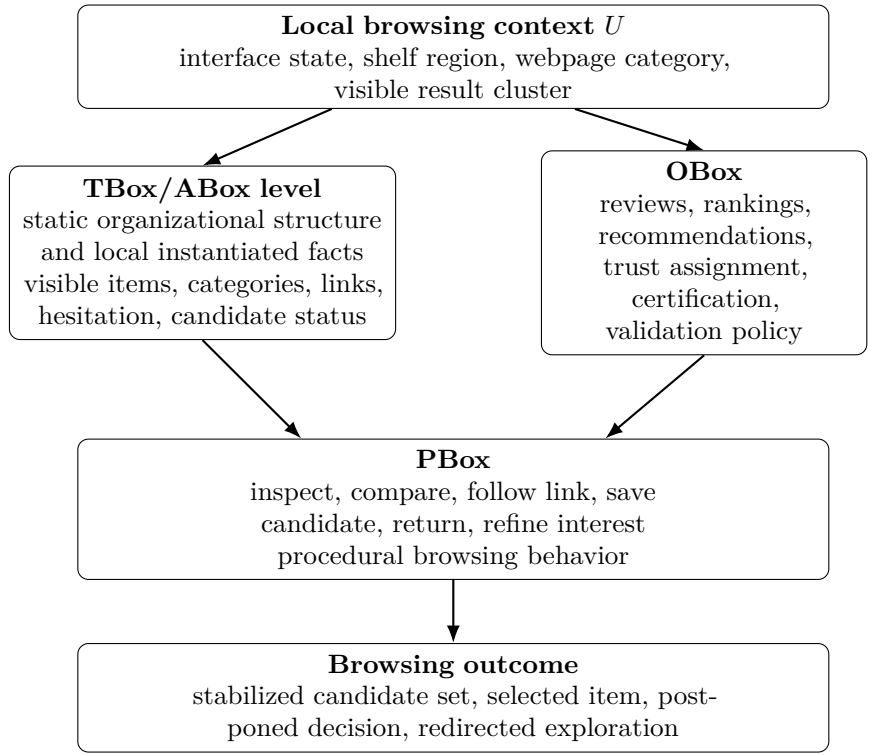

We now reinterpret central ingredients of browsing theory
within this structure.

\paragraph{TBox layer.}
The TBox may contain concepts such as
\[
\begin{aligned}
&\mathsf{Document},\ \mathsf{Category},\ \mathsf{VisibleItem},\ \mathsf{ProminentItem},\\
&\mathsf{CandidateItem},\ \mathsf{RelevantItem},\ \mathsf{TrustedCue},\ \mathsf{ReviewSignal},\\
&\mathsf{InterfaceElement}
\end{aligned}
\]
as well as roles such as
\[
\begin{aligned}
&\mathsf{displayedIn}(x,U),\ \mathsf{belongsTo}(x,C),\ \mathsf{linkedTo}(x,y),\\
&\mathsf{recommendedWith}(x,y),\ \mathsf{hasCue}(x,\kappa)
\end{aligned}
\]
At this level,
one captures the relatively stable classificatory and relational structure
within which browsing takes place.

\paragraph{ABox layer.}
The ABox records the current local informational facts,
for example
\[
d_1:\mathsf{VisibleItem}@U,\qquad
d_2:\mathsf{ProminentItem}@U,\qquad
(d_1,d_2):\mathsf{linkedTo}@U,
\]
or
\[
u:\mathsf{User}@U,\qquad
(u,d_1):\mathsf{inspects}@U,\qquad
(u,d_2):\mathsf{hesitatesOver}@U.
\]
Thus the ABox gives the presently instantiated browsing situation.

\paragraph{PBox layer.}
The PBox represents the procedural side of browsing.
This is crucial.
Browsing is typically not a single inference step from query to answer,
but a process involving actions such as:
\[
\mathsf{inspect},\
\mathsf{compare},\
\mathsf{scroll},\
\mathsf{followLink},\
\mathsf{openPreview},\
\mathsf{saveCandidate},\
\mathsf{return},\
\mathsf{refineInterest}.
\]
Accordingly, one may introduce guarded procedures such as
\[
\mathbf{if}\ \mathsf{prominent}(d)\ \mathbf{then}\ \mathsf{inspect}(d),
\]
\[
\mathbf{if}\ \mathsf{uncertain}(u)\ \mathbf{then}\ \mathsf{compare}(d_1,d_2),
\]
\[
\mathbf{if}\ \mathsf{candidate}(d)\ \mathbf{then}\ \mathsf{save}(d),
\]
\[
\mathbf{while}\ \mathsf{not\_settled}(u)\ \mathbf{do}\ \mathsf{browse\_next}.
\]
Here the symbols $\mathsf{prominent}(d)$, $\mathsf{uncertain}(u)$, $\mathsf{candidate}(d)$, and $\mathsf{not\_settled}(u)$ should be read as guard expressions evaluated by the metalevel guard judgment of \cref{sec:guards}. This layer is what makes it possible
to model browsing as an unfolding course of action
rather than as a single declarative fact.

\paragraph{OBox layer.}
The OBox represents the selective incorporation of external inputs.
In a browsing setting,
such external inputs include:
reviews, recommendations, popularity signals,
other users' annotations, staff suggestions,
or system-generated rankings.
These are not merely additional assertions.
They are externally sourced signals,
often heterogeneous in reliability,
and therefore naturally fall under the refined OBox:
query, response, trust assignment,
certification data, validation policy, and validated import.

In this perspective,
a browsing system may issue a query
\[
q=\text{``Is item $d$ widely regarded as useful for topic $T$?''}
\]
and receive one or more responses $\rho$.
Only after trust-sensitive validation
does the response become importable into the current knowledge state.
Thus one may obtain an update of the form
\[
(\State_U,\OBox_U)\vdash q \xRightarrow{V} \State_U',
\]
where the imported information is not raw testimony,
but already filtered through a validation policy.

\subsection{Formal decomposition of browsing dimensions}

Chang-style browsing theory is valuable partly because it does not treat browsing
as one-dimensional.
The present framework suggests the following analytic decomposition.

\begin{enumerate}
\item \emph{Contextual dimension.}
This corresponds to the local site $U$ and to the contextual assertions
carried by $\ABox_U$.
Examples include the current page,
the shelf segment under inspection,
or the visible cluster of results.

\item \emph{Behavioral dimension.}
This is captured by the PBox,
whose rules describe movement, comparison, postponement,
revisiting, and selection.

\item \emph{Motivational and cognitive dimension.}
This may be formalized by assertions such as
\[
u:\mathsf{Undecided}@U,\qquad
u:\mathsf{Curious}@U,\qquad
u:\mathsf{RiskAverse}@U,
\]
which then function as guards in the PBox.

\item \emph{Resource dimension.}
Available cues, recommendations, reviews,
or interface affordances may be represented
either as TBox/ABox-level objects or,
when externally supplied and selectively incorporated,
through the OBox.

\item \emph{Consequential dimension.}
The end result of a browsing episode
need not be a single retrieved answer.
It may instead be a transformed candidate set,
a revised preference ordering,
a narrowed region of attention,
or a decision to postpone commitment.
These are naturally represented
as resulting TAPO-states after procedural and oracle-mediated updates.
\end{enumerate}

The point of this decomposition is not terminological.
It shows that browsing theory can be made formally articulate:
different dimensions correspond to different logical layers,
and their interaction may be studied explicitly.

\subsection{A generic browsing episode in TAPO form}

We now sketch a generic browsing episode.
Suppose the user begins in a state
\[
X_U=(\State_U,\PBox_U,\OBox_U),
\]
where the ABox includes visible items
\[
d_1:\mathsf{VisibleItem}@U,\qquad d_2:\mathsf{VisibleItem}@U,
\]
and a hesitation assertion
\[
u:\mathsf{Undecided}@U.
\]

At the PBox level one may have a rule
\[
\mathbf{if}\ \mathsf{Undecided}(u)\ \mathbf{then}\ \mathsf{compare}(d_1,d_2).
\]
Suppose the corresponding guard judgment evaluates $\mathsf{Undecided}(u)$ as true. If the comparison produces no stable preference, so that the guard judgment also evaluates $\mathsf{still\_undecided}(u)$ as true, then the following procedural rule is activated:
\[
\mathbf{if}\ \mathsf{still\_undecided}(u)\ \mathbf{then}\ \mathsf{consult}(q_{d_1,d_2}).
\]
This is precisely where the OBox becomes relevant.
A query is issued, for example:
\[
q_{d_1,d_2}=\text{``Which of $d_1,d_2$ is more useful for topic $T$?''}
\]
The system or environment returns responses
\[
\rho_1,\rho_2,\dots
\]
with associated trust values and possibly certification data.
A validation policy $V$ selects which responses may be imported.
If the validated import yields
\[
d_1:\mathsf{RecommendedItem}@U,
\]
then a further PBox rule may update the candidate structure:
\[
\mathbf{if}\ d_1:\mathsf{RecommendedItem}@U
\ \mathbf{then}\ \mathsf{saveCandidate}(d_1).
\]
Thus the browsing process is not modeled
as an unstructured sequence of impressions,
but as a controlled alternation between
internal procedural moves and validated external enrichment.

\subsection{Why this goes beyond a descriptive restatement}

One may ask whether this appendix merely restates browsing theory
in formal notation.
The answer is no.
The gain of the TAPO perspective is at least fourfold.

\paragraph{First,} it distinguishes clearly between
static classificatory structure and procedural behavior.
Many descriptions of browsing intermix these.
In TAPO-description logic,
the distinction between TBox/ABox and P is explicit.

\paragraph{Second,} it separates internal browsing behavior
from externally supplied evaluative cues.
This is the contribution of the refined OBox.
Reviews, rankings, and recommendations do not simply become facts;
they become validated imports.

\paragraph{Third,} it makes comparison possible.
Different interface environments,
different validation policies,
or different user dispositions
can be represented by different TAPO-configurations,
and their consequences can then be compared.

\paragraph{Fourth,} it opens the way to proof-theoretic
and semantic analysis.
Once browsing is represented by judgments,
procedural transitions,
and oracle-sensitive imports,
one may ask not only \emph{what} happens,
but under which rules it follows,
whether certain outcomes are derivable,
and how those derivations correspond to the categorical semantics.

\subsection{Relation to interface-sensitive browsing}

This formalization is especially suggestive
for interface-sensitive studies of browsing.
If item prominence, grouping, labeling,
or recommendation widgets affect user behavior,
then such influence may be represented at two levels.
At the TBox/ABox level,
the interface provides visible organizational structure.
At the PBox level,
this structure induces or biases procedures of inspection and navigation.
At the OBox level,
externally generated ratings, reviews, or recommendation signals
enter the process selectively.
Hence the present framework may be viewed
as a mathematical apparatus
for analyzing how organization, interface, and external cues
jointly shape browsing behavior.

\subsection{Toward a formal research program}

The foregoing observations suggest that TAPO-description logic
may support an independent research program
on the formal analysis of browsing theory.
At least the following questions appear natural:

\begin{enumerate}
\item Which classes of browsing episode can be represented
by finite PBox procedures?

\item Under what conditions do repeated browsing and external consultation
stabilize to a fixed candidate set?

\item How does one compare two interface environments
in terms of the procedural and oracle-sensitive paths
they induce?

\item Can one define formal notions of browsing success,
serendipitous discovery,
or premature abandonment
within TAPO-structured knowledge dynamics?

\item To what extent can local browsing episodes
be glued into a larger global account
of information behavior via sheaf-theoretic methods?
\end{enumerate}

These questions go beyond the immediate scope of the present paper,
but they indicate that the framework developed here
is not merely compatible with browsing theory:
it can also serve as a basis
for a more formal and comparative theory of browsing itself.

\begin{remark}
From this perspective,
the present paper may be read in two complementary ways.
On the one hand,
it develops TAPO-description logic
as a logic for information behavior.
On the other hand,
it suggests that established theories of browsing,
especially multidimensional and interface-sensitive ones,
may be reinterpreted within TAPO-description logic
as formally analyzable systems of local knowledge states,
procedural transitions, and validated external inputs.
\end{remark}

\end{document}